\documentclass[journal]{IEEEtran}

\usepackage{amsmath,amssymb,amsfonts}
\usepackage{graphicx}
\graphicspath{{figs/}}
\usepackage{hyperref}
\usepackage{cite}
\usepackage{bm}
\usepackage{array}
\usepackage{multirow}
\usepackage{float}
\usepackage{textcomp}
\usepackage{subcaption}

\newtheorem{theorem}{Theorem}
\newtheorem{proposition}{Proposition}
\newtheorem{corollary}{Corollary}
\newtheorem{remark}{Remark}

\begin{document}

\title{CisLunarSense: Opportunistic ISAC for Debris Detection at the Lunar Gateway}

\author{Haofan~Dong,~\IEEEmembership{Student~Member,~IEEE,}
        and~Ozgur~B.~Akan,~\IEEEmembership{Fellow,~IEEE}
\thanks{H.\ Dong and O.\ B.\ Akan are with the Internet of Everything Group, 
Department of Engineering, University of Cambridge, CB3 0FA Cambridge, U.K.(e-mail: hd489@cam.ac.uk, oba21@cam.ac.uk)}
\thanks{O.\ B.\ Akan is also with the Center for neXt-generation 
Communications (CXC), Department of Electrical and Electronics Engineering, 
Ko\c{c} University, 34450 Istanbul, Turkey.}}

\maketitle

\begin{abstract}
We propose CisLunarSense, an opportunistic integrated 
sensing and communication (ISAC) framework that exploits the 
Lunar Gateway's Ka-band relay for monostatic debris detection, 
addressing the absence of cislunar space situational awareness 
infrastructure beyond the reach of ground-based radars.
Using NASA/ESA-documented system parameters with 
author-selected sensing settings and a CR3BP-based 9:2 
near-rectilinear halo orbit model, we derive the
orbit-phase-dependent Cram\'{e}r--Rao bound under OFDM 
inter-carrier interference, quantify a 36~dB cislunar sensing 
advantage over a ground-based Ka-band reference, and design a 
velocity-adaptive processor with mode switching at 337~m/s.
Gateway operational debris ($v_\mathrm{rel} < 50$~m/s) is 
detectable within 700~km with over 30~minutes of warning; 
external threats ($v_\mathrm{rel}$ up to 500~m/s) remain 
detectable within 400--630~km.
An orbit-phase-adaptive allocation reduces the sensing duty 
cycle from 60\% to 19\%, increasing relay throughput from 
44 to 90~Mbps.
A closed-form sensing outage probability for $K$-CPI 
non-coherent integration under Swerling~I fluctuation 
shows that the 10\%-outage detection range reaches 91\% 
of the deterministic maximum at the nominal operating 
point $K = 16$.
\end{abstract}


\section{Introduction}
\label{sec:intro}

The Lunar Gateway, scheduled for NRHO assembly around the 
Earth--Moon $L_2$ point beginning in 2028 \cite{WhiteHouse2022}, 
faces a debris hazard with no existing surveillance capability.
Objects separated during routine operations remain in the NRHO 
vicinity for 30--100 days, repeatedly approaching at 
30--50~m/s \cite{Davis2019,Scheuerle2023}, beyond the reach 
of ground-based radars at cislunar distances 
\cite{NASARadar2023,Barakat2024}.

Integrated sensing and communication (ISAC) 
\cite{LiuFan2022,ZhangCui2021}, where a dual-function waveform 
carries data and enables radar sensing, offers a potential 
solution.
OFDM-based ISAC is a leading 6G candidate 
\cite{ChiragSurvey2022,SturmWiesbeck2011}.
Dual-function radar-communication (DFRC) systems have been 
studied extensively for terrestrial applications 
\cite{Chiriyath2017,Hassanien2016,Mishra2019};
space extensions include THz ISAC for LEO 
\cite{DongAkan2025_TWC}, bistatic LEO detection 
\cite{Anttonen2021}, and rate-splitting ISAC 
\cite{Park2024}, none of which address cislunar channels 
or debris trajectory modeling 
(Table~\ref{tab:positioning}).

\begin{table}[!t]
\centering
\caption{Positioning Relative to Related Work}
\label{tab:positioning}
\renewcommand{\arraystretch}{1.1}
\begin{tabular}{lcccccc}
\hline
\textbf{Work} & \textbf{Domain} & \textbf{Band} & 
\textbf{CRB} & \textbf{Orbit} & \textbf{Debris} & \textbf{S-C Alloc.} \\
\hline
\cite{LiuFan2022} & Terr. & Sub-6 & \checkmark & -- & -- & -- \\
\cite{ChiragSurvey2022} & Terr. & LTE/NR & -- & -- & -- & -- \\
\cite{DongAkan2025_TWC} & LEO & THz & -- & Kepler & -- & -- \\
\cite{Anttonen2021} & LEO & S & -- & Kepler & -- & -- \\
\cite{Park2024} & LEO & -- & -- & Kepler & -- & -- \\
\textbf{This work} & \textbf{Cislunar} & \textbf{Ka} & 
\checkmark & \textbf{CR3BP} & \checkmark & \checkmark \\
\hline
\end{tabular}
\end{table}

Cislunar ISAC differs from prior work in two respects: 
(i)~the free-space channel provides a 36~dB sensing advantage 
over a ground-based Ka-band reference 
(Section~\ref{sec:system_model}); 
(ii)~the NRHO eccentricity ($e = 0.91$) creates a 
$100\times$ velocity variation (0.07--1.7~km/s), coupling 
with ICI to make sensing performance orbit-phase-dependent.

This paper proposes CisLunarSense, exploiting the Gateway's 
HLCS Ka-band relay \cite{Dellandrea2022,Johnson2021} for 
opportunistic monostatic debris detection 
\emph{without dedicated radar hardware or spectrum, assuming 
multicarrier-capable relay evolution}.
While orbiter-class bistatic radar has been demonstrated for 
lunar surface science (e.g., LRO Mini-RF, Chang'e bistatic 
experiments), to our knowledge no prior work addresses 
OFDM-ISAC for cislunar debris detection under the unique 
velocity-coupling constraints of an NRHO. 
The contributions are:
\begin{enumerate}
\item \emph{Cislunar ISAC fundamentals:} 
We identify the NRHO-specific coupling between a $100\times$
intra-orbit velocity sweep and OFDM ICI as a new sensing
regime---absent from LEO (nearly constant $v_\mathrm{rel}$)
and GEO ($v_\mathrm{rel}\!\approx\!0$) ISAC literature---and
(Theorems~\ref{thm:fim}--\ref{thm:crb}) and Swerling~I 
detection range (Theorem~\ref{thm:rmax}).

\item \emph{Adaptive processing and scheduling:} 
velocity-adaptive 2D-FFT/Keystone processor with mode switching 
at 337~m/s; sequential sensing--communication allocation 
(Proposition~\ref{prop:alloc}); orbit-phase-adaptive refinement 
(Proposition~\ref{prop:adaptive}) reducing the duty cycle from 
60\% to 19\% while maintaining detection coverage.

\item \emph{CR3BP-validated evaluation:} 
debris recontact trajectories validated against NASA studies 
\cite{Davis2019,Scheuerle2023} yield 700~km detection for 
operational debris and 400--630~km for external threats.
Monte Carlo validation confirms the detection model 
(Fig.~\ref{fig:mc_pd}), and a closed-form $K$-CPI sensing 
outage probability (Proposition~\ref{prop:outage}) 
quantifies the operational reliable range as a function of 
integration depth (Fig.~\ref{fig:outage}).
\end{enumerate}


Sections~\ref{sec:system_model}--\ref{sec:design} develop the 
system model, fundamental limits, and adaptive design.
Section~\ref{sec:setup} describes the CR3BP debris model, and 
Sections~\ref{sec:results}--\ref{sec:conclusion} present 
results and conclusions.

\begin{figure*}[!t]
\centering
\includegraphics[width=0.85\textwidth]{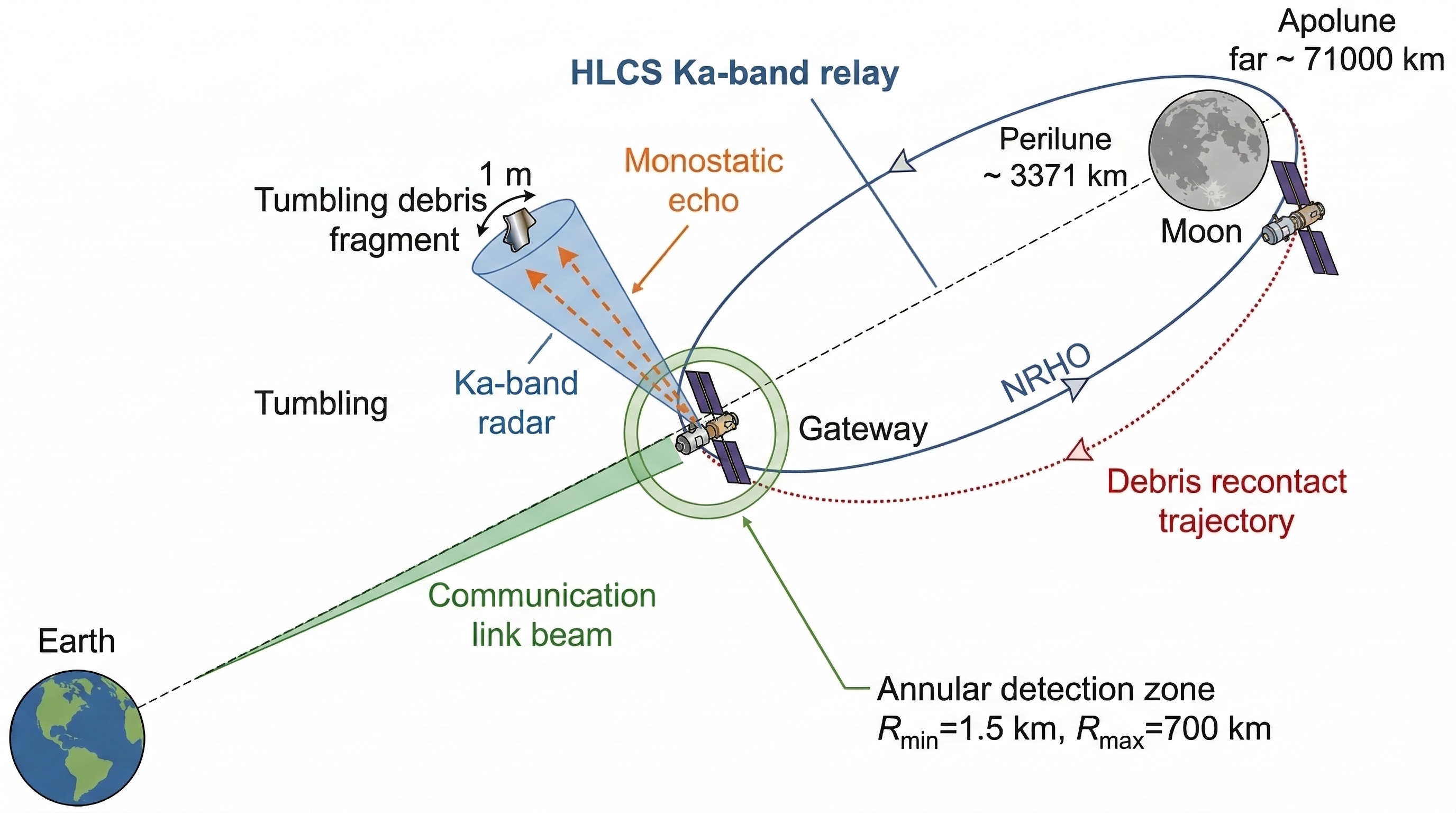}
\caption{CisLunarSense system overview. HLCS Ka-band relay 
beam (green) for Earth uplink; echoes (orange) from debris 
within $R_\mathrm{min}=1.5$~km to $R_\mathrm{max}=700$~km 
(1~m target) are processed for monostatic detection. Red 
dotted: CR3BP-predicted recontact trajectory.}
\label{fig:system}
\end{figure*}

\section{System Model}
\label{sec:system_model}

We consider the Lunar Gateway operating in a 9:2 synodic resonant 
near-rectilinear halo orbit (NRHO) about the Earth--Moon $L_2$ point.
The HALO Lunar Communication System (HLCS), built by Thales Alenia Space
for ESA, serves as the primary communication relay between the lunar surface
and Earth \cite{Dellandrea2022}.
CisLunarSense exploits the HLCS Ka-band relay for opportunistic 
monostatic debris sensing during active relay sessions, requiring 
only a multicarrier-capable transceiver and allocating a 
configurable fraction of session time to sensing.

\subsection{OFDM-ISAC Signal Model}

The HLCS transmits an OFDM waveform at $f_c = 27$~GHz 
($\lambda = 11.1$~mm) with bandwidth $B = 100$~MHz, $N$ 
subcarriers indexed by 
$n \in \{-(N-1)/2, \ldots, (N-1)/2\}$ centered around the 
carrier, $M$ symbols per CPI, subcarrier spacing 
$\Delta f = B/N$, symbol duration 
$T_\mathrm{sym} = (1 + 1/8)/\Delta f$ (CP ratio $1/8$), 
and CPI duration $T_\mathrm{CPI} = M T_\mathrm{sym}$.
Data symbols $d[n,m] \in \mathcal{A}$ use the OQPSK 
constellation specified for HLCS \cite{Dellandrea2022}.
While the current HLCS baseline employs single-carrier 
modulation, the analysis adopts OFDM consistent with 
proposed Ka-band relay upgrades \cite{Johnson2021}; the 
framework is forward-compatible with the operational HLCS 
upon adoption of multicarrier modulation.
The OFDM assumption enables 2D-FFT-based joint 
range-Doppler processing 
\cite{LiuFan2022,Gaudio2020,Liyanaarachchi2021,Mishra2019}.

For a monostatic radar receiver co-located with the HLCS transmitter,
the echo from a point target at range $R$ with radial velocity $v$ is
\begin{equation}
y[n,m] = \alpha \, d[n,m] \, e^{-j2\pi n \Delta f \tau} 
e^{j2\pi m T_\mathrm{sym} f_d} + w[n,m]
\label{eq:rx}
\end{equation}
where $\tau = 2R/c$ is the round-trip delay, 
$f_d = 2v f_c/c$ is the two-way Doppler shift, 
$w[n,m] \sim \mathcal{CN}(0,\sigma_w^2)$ is additive white Gaussian 
noise with $\sigma_w^2 = k_B T_\mathrm{sys} \Delta f$, and 
$\alpha$ is the complex channel coefficient satisfying
\begin{equation}
|\alpha|^2 = \frac{P_t G_t G_r \lambda^2 \sigma}{(4\pi)^3 R^4}
\label{eq:alpha}
\end{equation}
with $P_t$ denoting the transmit power, $G_t$ and $G_r$ the transmit
and receive antenna gains, and $\sigma$ the target radar cross section
(RCS).

Since the transmitted data $d[n,m]$ is known to the co-located 
receiver, element-wise division yields
\begin{equation}
\tilde{y}[n,m] = \alpha \, e^{-j2\pi n \Delta f \tau} 
e^{j2\pi m T_\mathrm{sym} f_d} + \tilde{w}[n,m]
\label{eq:rx_clean}
\end{equation}
where $\tilde{w}[n,m] \sim \mathcal{CN}(0,\sigma_w^2/P_s)$, and
$P_s = |d[n,m]|^2$ is constant for the OQPSK constellation.

\begin{remark}[Monostatic TDD Operation]
\label{rem:sic}
At $R = 100$~km the echo is 176~dB below the transmit power, 
far exceeding Ka-band self-interference cancellation limits 
(65--100~dB) \cite{Lagunas2025}; full-duplex is therefore 
infeasible.
We adopt TDD consistent with the HLCS link structure: each 
relay session alternates between communication transmission 
bursts and dedicated sensing receive intervals of total 
duration $\rho T_\mathrm{obs}$ 
(Proposition~\ref{prop:alloc}).
During each sensing interval, the transmitter is silent 
and the receiver captures echoes of the immediately 
preceding transmission burst.
The TX--RX guard interval ($\sim$10~$\mu$s) sets the minimum 
detectable range $R_\mathrm{min} \approx 1.5$~km; 
$R_\mathrm{max}$ is determined by the link budget 
\eqref{eq:rmax}.
\end{remark}

\subsection{Cislunar Channel Model}

The cislunar channel differs from terrestrial and LEO scenarios 
in three respects: (i)~pure free-space propagation with $R^{-4}$ 
path loss and no atmospheric or ionospheric attenuation; 
(ii)~a low system noise temperature 
$T_\mathrm{sys} \approx 200$~K, set by the CMB (2.725~K) and a 
Ka-band LNA with 2.0~dB noise figure \cite{KaBandLNA2001}, 
versus $\sim$290~K for ground-based receivers; and 
(iii)~zero ground clutter, which in LEO can raise the effective 
noise floor by 20--30~dB.

\begin{remark}[Cislunar Environmental Advantage]
\label{rem:advantage}
Relative to a ground-based Ka-band radar attempting the same 
detection task---the operational baseline for Earth-based SSA 
that fails at cislunar distances 
\cite{NASARadar2023,Barakat2024}---the combined effect of 
these three factors provides an advantage of approximately 
36~dB, decomposed as: atmospheric absorption 
($2 \times 3.0 = 6.0$~dB two-way), ionospheric scintillation 
($2 \times 1.5 = 3.0$~dB two-way), ground clutter (25.0~dB), 
and thermal noise reduction 
($10\log_{10}(290/200) = 1.6$~dB). 
Under the $R^4$ path loss law, this translates to a $7.8\times$ 
improvement in maximum detection range over a co-located 
ground-based reference system. A space-to-space LEO ISAC 
configuration would already capture the atmospheric, 
ionospheric, and clutter terms; the residual cislunar advantage 
over such a configuration arises primarily from the absence of 
RF traffic congestion and is on the order of a few dB.
\end{remark}

For targets of interest ($d \geq 0.3$~m), $x = \pi d/\lambda 
\geq 85$ places scattering in the optical regime with 
$\sigma \approx \pi(d/2)^2$ \cite{XuKennedy2019}; the NASA 
SEM gives $\pm 4$~dB (1$\sigma$) RCS uncertainty (sensitivity 
in Section~VI). Tumbling debris is modeled as Swerling~I as a conservative 
single-parameter baseline: the slow-fluctuation condition 
$T_\mathrm{CPI}\!\ll\!1/(2 f_\mathrm{tumble})$ holds with 
$>\!10\times$ margin for nominal CPI durations and 
0.1--10~rpm tumbling \cite{Hejduk2010,Richards2014}, 
justifying constant RCS within a CPI; class-I (versus 
class-III) is adopted to avoid optimistic detection-threshold 
predictions for morphologies with multiple comparable scatterers. Three 
modeling simplifications warrant note: lunar/Earth-limb 
backscatter is geometrically excluded by NRHO altitude and 
$>\!20^\circ$ off-limb antenna pointing for $R\!>\!1000$~km 
targets, and is suppressed below the noise floor in the 
$400$--$1000$~km headline regime by the antenna sidelobe 
envelope ($>\!30$~dB below boresight at $>\!2\,\theta_\mathrm{bw}$ 
off-axis); spacecraft-body multipath is reduced (but not 
eliminated) by HLCS two-axis gimbal pointing within allowed 
structural exclusion sectors; detailed clutter/multipath 
modeling is left to mission-stage analysis.

\subsection{NRHO Orbit-Phase Parameterization}
\label{sec:nrho}

The Gateway's 9:2 NRHO has perilune radius $r_p = 3{,}371$~km and 
apolune radius $r_a = 71{,}476$~km, with an orbital period of 
$T_\mathrm{orb} = 6.62$~days \cite{Lee2019}.
We parameterize the orbit by phase $\theta \in [0, 2\pi)$, with
$\theta = 0$ at apolune and $\theta = \pi$ at perilune.
Since the NRHO is a three-body orbit that does not follow 
Keplerian conic sections, the Gateway velocity profile 
$v_\mathrm{GW}(\theta)$ is obtained by numerical integration 
of the CR3BP equations of motion (Section~\ref{sec:cr3bp}) 
using the same reference trajectory employed for debris 
propagation.
The resulting velocity varies from 
$v_\mathrm{GW} \approx 72$~m/s at apolune to 
$v_\mathrm{GW} \approx 1{,}673$~m/s at perilune, 
with the transition concentrated in a narrow angular window 
($\sim\!30^\circ$) around perilune.
The Keplerian velocity deviates from the CR3BP profile 
by up to $5\times$ at intermediate orbital phases; all 
numerical results use the CR3BP velocity exclusively.

The CR3BP reference trajectory confirms strong apolune residence: 
the Gateway spends approximately 85--90\% of each orbital period 
in the low-velocity regime near apolune.
This asymmetry has a direct consequence for opportunistic sensing:
the majority of observation time occurs in the low-velocity regime
where OFDM processing gain is maximized.

\subsection{HLCS System Parameters}
\label{sec:hlcs_params}

The HLCS carries two 1.25~m diameter Ka-band antennas on two-axis
gimbals, each driven by a 35~W solid-state power amplifier 
\cite{Dellandrea2022}. The resulting system parameters are 
summarized in Table~\ref{tab:params}.

\begin{table}[!t]
\centering
\caption{HLCS System Parameters for Sensing Analysis}
\label{tab:params}
\begin{tabular}{lll}
\hline
\textbf{Parameter} & \textbf{Value} & \textbf{Source} \\
\hline
Carrier frequency $f_c$ & 27 GHz & \cite{Dellandrea2022} \\
Wavelength $\lambda$ & 11.1 mm & \\
Bandwidth $B$ & 100 MHz & \cite{Johnson2021} \\
Antenna diameter $D$ & 1.25 m & \cite{Dellandrea2022} \\
Antenna gain $G_t = G_r$ & 48.4 dBi & Calculated ($\eta = 0.55$) \\
TX power $P_t$ & 35 W (15.4 dBW) & \cite{Dellandrea2022} \\
EIRP & 63.8 dBW & \\
System temperature $T_\mathrm{sys}$ & 200 K & \cite{KaBandLNA2001} \\
Subcarriers $N$ & 1024 & \\
Symbols/CPI $M$ & 64 & \\
Subcarrier spacing $\Delta f$ & 97.7 kHz & \\
CPI duration $T_\mathrm{CPI}$ & 0.74 ms & \\
\hline
\end{tabular}
\end{table}
The Gateway--Earth relay link distance varies from 
$d_\mathrm{peri} \approx 381{,}000$~km at perilune 
(Gateway nearest to Earth) to 
$d_\mathrm{apo} \approx 456{,}000$~km at apolune, 
producing a 1.6~dB one-way path loss variation.
The corresponding relay throughput 
$R_\mathrm{full}(\theta)$ ranges from approximately 
104~Mbps at apolune to 116~Mbps at perilune, 
motivating the orbit-phase-adaptive allocation in 
Proposition~\ref{prop:adaptive}.

The HLCS operates on demand during lunar relay sessions, with a 
default standby mode when no surface assets require service
\cite{Dellandrea2022}. The proposed sensing function is therefore
\emph{opportunistic}: it exploits echoes of communication waveforms
that are transmitted for relay purposes, adding no dedicated radar hardware or spectrum while 
maintaining the minimum relay throughput 
(Section~\ref{sec:allocation}).

\section{Orbit-Phase-Dependent Fundamental Limits}
\label{sec:crb}

\begin{figure*}[!t]
\centering
\includegraphics[width=0.92\textwidth]{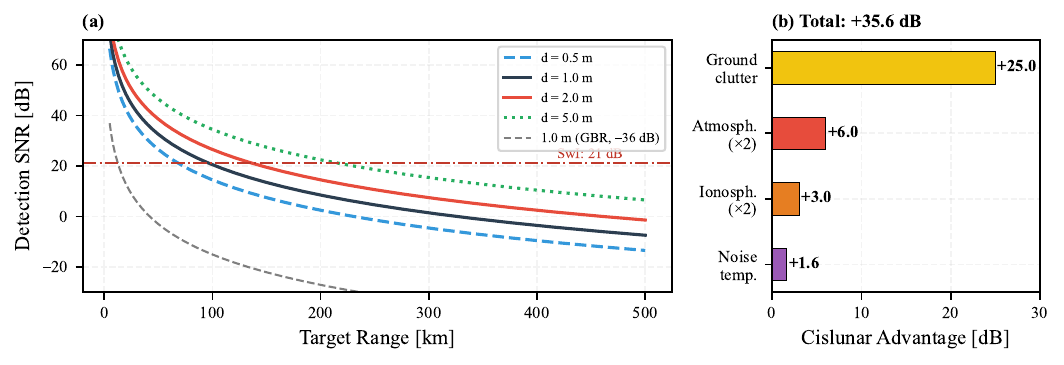}
\caption{(a)~Detection SNR vs target range at apolune 
($K=16$) for four debris diameters and a ground-based 
Ka-band reference ($-36$~dB penalty, see panel b). 
Swerling~I threshold: 21.1~dB. (b)~Decomposition of the 
35.6~dB cislunar advantage.}
\label{fig:snr_advantage}
\end{figure*}

\subsection{Sensing SNR}

\begin{figure*}[!t]
\centering
\includegraphics[width=0.92\textwidth]{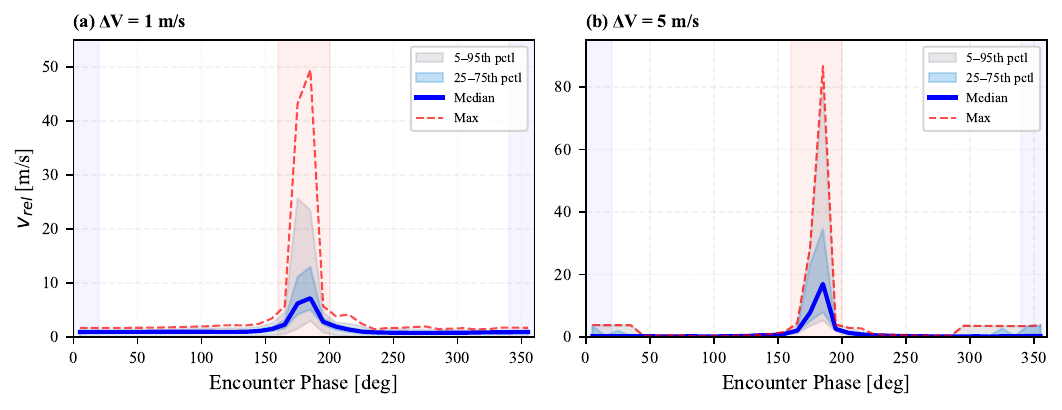}
\caption{CR3BP-derived debris recontact velocity as a function of encounter orbital phase ($0^\circ$ = apolune, $180^\circ$ = perilune) for (a) $\Delta v = 1$~m/s and (b) $\Delta v = 5$~m/s separations. Shaded bands show 5th--95th and 25th--75th percentile ranges; dashed line is the maximum observed value.}
\label{fig:vrel_bands}
\end{figure*}

\begin{figure*}[!t]
\centering
\includegraphics[width=0.85\textwidth]{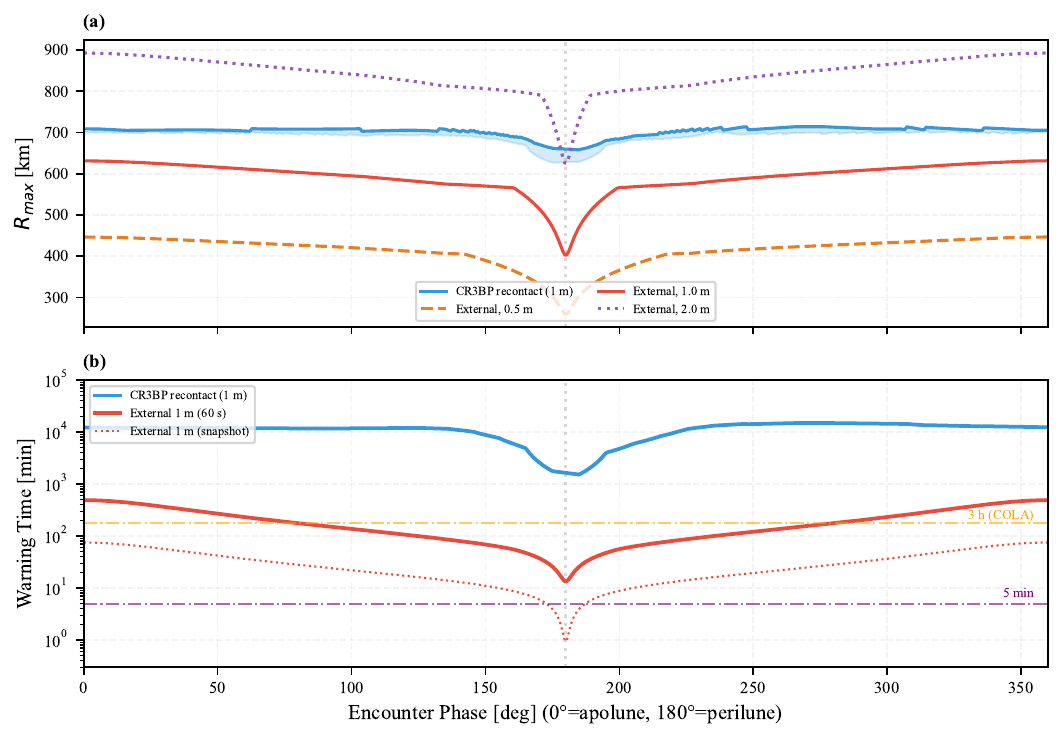}
\caption{(a) Maximum detection range along the NRHO for a 1~m target (CR3BP recontact, blue with uncertainty band) and parametric external-threat stress tests at three target sizes ($v_\mathrm{rel} = 0.3 \times v_\mathrm{GW}$). (b) Proximity warning time with 3-hour COLA and 5-minute emergency thresholds. The snapshot baseline ($K = 16$, dotted) falls below 1~minute near perilune.}
\label{fig:rmax_warning}
\end{figure*}

The total sensing SNR after 2D-FFT processing of one CPI 
follows from substituting the per-element signal power 
$|\alpha|^2$ from \eqref{eq:alpha} and the per-subcarrier 
noise variance $\sigma_w^2 = k_B T_\mathrm{sys}\Delta f$ 
into $\gamma = |\alpha|^2 \cdot NM \cdot \mathcal{G}/\sigma_w^2$. 
The coherent gain $NM$ multiplies the signal power while 
the noise bandwidth folds back to the full $B = N\Delta f$, 
yielding
\begin{equation}
\gamma(\theta, R, \sigma) = 
\frac{\mathrm{EIRP} \cdot G_r \cdot \lambda^2 \cdot \sigma \cdot 
NM \cdot \mathcal{G}(v_\mathrm{rel}(\theta))}
{(4\pi)^3 \cdot R^4 \cdot k_B T_\mathrm{sys} B}
\label{eq:snr}
\end{equation}
where $\mathrm{EIRP} = P_t G_t$, and 
$\mathcal{G}(v_\mathrm{rel})$ is the ICI-aware processing 
efficiency defined in Proposition~\ref{prop:ici}. The relation 
to the per-subcarrier SNR used in Theorem~\ref{thm:fim} is 
$\gamma = NM\,\mathcal{G}\,\gamma_\mathrm{sc}$, where the 
factor $N$ reflects coherent integration across subcarriers.

\begin{proposition}[ICI-Aware Processing Efficiency]
\label{prop:ici}
For a target with radial velocity $v_\mathrm{rel}$ at carrier 
frequency $f_c$, the coherent processing efficiency relative to the 
ideal gain $NM$ is
\begin{equation}
\mathcal{G}(v_\mathrm{rel}) = 
\mathrm{sinc}^2\!\left(\frac{2 v_\mathrm{rel} f_c}{c \Delta f}\right) 
\cdot \eta_\mathrm{impl}
\label{eq:G_proc}
\end{equation}
where $\mathrm{sinc}(x) \triangleq \sin(\pi x)/(\pi x)$ and 
$\eta_\mathrm{impl} \in (0,1]$ accounts for windowing and 
implementation losses.
\end{proposition}
Throughout this paper, $\eta_\mathrm{impl} = 0.56$ ($-2.5$~dB), 
accounting for Hamming window processing loss ($-1.4$~dB) and 
digital implementation margin ($-1.1$~dB) \cite{Richards2014}.

\begin{IEEEproof}
The ICI power loss follows from the DFT response of an OFDM 
receiver to a Doppler-shifted input signal. The useful signal power
at the target subcarrier is attenuated by the factor 
$|\mathrm{sinc}(f_d/\Delta f)|^2$, where $f_d = 2v_\mathrm{rel}f_c/c$
is the two-way Doppler shift \cite{Braun2014}.
The remaining power $1 - \mathrm{sinc}^2(f_d/\Delta f)$ 
redistributes across adjacent subcarriers as interference.
This model assumes ideal time-frequency synchronization, 
negligible phase noise, and range migration within one 
resolution cell per CPI; the last condition is ensured by 
Proposition~\ref{prop:cpi} for the HLCS parameters in 
Table~\ref{tab:params}, while the first two are standard 
assumptions for Ka-band space relay links \cite{Richards2014}.
\end{IEEEproof}

The efficiency $\mathcal{G}$ depends on orbital phase through 
$v_\mathrm{rel}(\theta)$. 
For co-orbital debris at apolune ($v_\mathrm{rel} \approx 10$~m/s),
$f_d/\Delta f \approx 0.018$ and $\mathcal{G} \approx \eta_\mathrm{impl}$
(negligible ICI).
For transfer-orbit debris at perilune 
($v_\mathrm{rel} \approx 500$~m/s),
$f_d/\Delta f \approx 0.92$ and $\mathcal{G}$ drops by over 20~dB.

\subsection{Fisher Information and Cram\'{e}r-Rao Bounds}

The parameter vector of interest is 
$\boldsymbol{\xi} = [\tau, f_d]^T$.
After data-symbol removal \eqref{eq:rx_clean}, the observation
$\tilde{\mathbf{y}} \in \mathbb{C}^{NM}$ follows a complex 
Gaussian distribution conditioned on $\boldsymbol{\xi}$. The 
amplitude $\alpha$ is absorbed into the per-subcarrier SNR 
$\gamma_\mathrm{sc} = |\alpha|^2 P_s / \sigma_w^2$ and does 
not appear as a separate nuisance parameter in the FIM; the 
score vectors for $\tau$ and $f_d$ are orthogonal to 
$\partial\mathbf{s}/\partial|\alpha|$ by the centered-index 
structure of the steering vector.

\begin{theorem}[Cislunar ISAC Fisher Information Matrix]
\label{thm:fim}
Under the signal model \eqref{eq:rx_clean} with known data 
symbols and Gaussian noise, the FIM for $(\tau, f_d)$ estimation
from a single CPI is
\begin{equation}
\mathbf{J}(\tau, f_d) = 2NM\gamma_\mathrm{sc}
\begin{bmatrix}
4\pi^2 \beta_\mathrm{rms}^2 & 0 \\
0 & 4\pi^2 \gamma_\mathrm{rms}^2
\end{bmatrix}
\label{eq:fim}
\end{equation}
where $\gamma_\mathrm{sc} = |\alpha|^2 P_s / \sigma_w^2$ is the 
per-subcarrier SNR, and
\begin{align}
\beta_\mathrm{rms}^2 &= \frac{N^2 - 1}{12}\,\Delta f^2 
\approx \frac{B^2}{12}
\label{eq:beta} \\
\gamma_\mathrm{rms}^2 &= \frac{M^2 - 1}{12}\,T_\mathrm{sym}^2 
\approx \frac{T_\mathrm{CPI}^2}{12}
\label{eq:gamma}
\end{align}
are the squared RMS bandwidth and squared RMS observation time,
respectively. The off-diagonal terms vanish due to the separable
structure of the delay-Doppler steering vector.
\end{theorem}

\begin{IEEEproof}
Applying Slepian--Bangs to 
$\tilde{\mathbf{y}}\sim\mathcal{CN}(\alpha\,\mathbf{a}(\tau,f_d),
(\sigma_w^2/P_s)\mathbf{I})$ with separable steering vector 
$[\mathbf{a}]_{n+mN}=e^{-j2\pi n\Delta f\tau}
e^{j2\pi m T_\mathrm{sym}f_d}$ over centered indices, the 
diagonal terms reduce to $\sum n^2=N(N^2-1)/12$ and 
$\sum m^2=M(M^2-1)/12$, yielding 
$[\mathbf{J}]_{\tau\tau}=8\pi^2 NM\gamma_\mathrm{sc}\beta_\mathrm{rms}^2$ 
and similarly for $f_d$. The cross-term $\sum_n n\cdot\sum_m m$ 
vanishes under centered indexing.
\end{IEEEproof}

\begin{theorem}[Orbit-Phase-Dependent Achievable Bound]
\label{thm:crb}
For the standard 2D-FFT radar processor with effective output SNR
$\gamma(\theta,R,\sigma)$ from \eqref{eq:snr}, the CRBs for range 
and velocity estimation at orbital phase $\theta$ are
\begin{align}
\mathrm{CRB}_R(\theta) &= 
\frac{c^2}{32\pi^2 \gamma(\theta,R,\sigma) \cdot \beta_\mathrm{rms}^2}
\label{eq:crb_R} \\
\mathrm{CRB}_v(\theta) &= 
\frac{c^2}{32\pi^2 f_c^2 \cdot \gamma(\theta,R,\sigma) \cdot 
\gamma_\mathrm{rms}^2}
\label{eq:crb_v}
\end{align}
where $\gamma(\theta,R,\sigma)$ is given by \eqref{eq:snr}.
Both CRBs depend on orbital phase $\theta$ through the 
ICI-aware processing efficiency $\mathcal{G}(v_\mathrm{rel}(\theta))$.
\end{theorem}

\begin{remark}[Processor-Specific vs.\ Information-Theoretic Bound]
\label{rem:crb_practical}
Equations \eqref{eq:crb_R}--\eqref{eq:crb_v} give the 2D-FFT 
processor's achievable bound, with the orbit-phase dependence 
entering through $\mathcal{G}(v_\mathrm{rel}(\theta))$. The 
information-theoretic CRB from Theorem~\ref{thm:fim} carries 
no $\mathcal{G}$ factor (the FIM is invariant to 
post-processing) and is therefore tighter than 
\eqref{eq:crb_R}--\eqref{eq:crb_v} by a factor of 
$1/\mathcal{G}$. In the low-velocity regime where 
$\mathcal{G}\approx\eta_\mathrm{impl}\approx 0.56$, this 
amounts to a fixed implementation loss of approximately 
$2.5$~dB; at high velocity the gap widens monotonically. 
For a 1~m target at 100~km with $v_\mathrm{rel}=10$~m/s and 
$K=16$, the resulting range RMSE floor is below 1~m.
\end{remark}

\subsection{Maximum Sensing Range}

\begin{theorem}[Feasibility Region]
\label{thm:rmax}
For a Swerling~I fluctuating target with RCS $\sigma$, the maximum
monostatic sensing range achieving detection probability $P_d$ at 
false alarm rate $P_{fa}$ with $K$ non-coherently integrated CPIs is
\begin{equation}
R_\mathrm{max}(\theta,\sigma,K) = 
\left[\frac{\mathrm{EIRP} \cdot G_r \cdot \lambda^2 \cdot \sigma 
\cdot NM \cdot \mathcal{G}(v_\mathrm{rel}(\theta)) \cdot K^\kappa}
{(4\pi)^3 \cdot k_B T_\mathrm{sys} B \cdot 
\gamma_\mathrm{th}}\right]^{\!1/4}
\label{eq:rmax}
\end{equation}
where $\gamma_\mathrm{th} = \ln(P_{fa})/\ln(P_d) - 1$ is the 
Swerling~I detection threshold (21.1~dB for $P_d = 0.9$, 
$P_{fa} = 10^{-6}$) and $\kappa \approx 0.85$ is the 
non-coherent integration efficiency exponent.
\end{theorem}

\begin{IEEEproof}
For a single CPI ($K = 1$), the Swerling~I detection probability 
satisfies $P_d = P_{fa}^{1/(1+\bar{\gamma})}$, where
$\bar{\gamma}$ is the average SNR \cite{Richards2014}.
Solving for $\bar{\gamma}$ gives the threshold 
$\gamma_\mathrm{th} = \ln P_{fa}/\ln P_d - 1$.
For $K>1$, the integration gain follows $K^\kappa$ with 
$\kappa\approx 0.85$, obtained by least-squares fit over 
$K\in[1,500]$ at $(P_d,P_\mathrm{fa})=(0.9,10^{-6})$ against 
the exact incomplete-gamma form \cite[Ch.~7]{Richards2014} 
(within 1~dB). Substituting into \eqref{eq:snr} and solving 
for $R$ yields \eqref{eq:rmax}.
\end{IEEEproof}

\begin{corollary}[Detection Range Summary]
\label{cor:range}
Table~\ref{tab:rmax} lists the maximum detection range for 
representative debris targets at apolune ($v_\mathrm{rel} = 10$~m/s)
with $K = 16$ non-coherently integrated CPIs.
\end{corollary}

\begin{table}[!t]
\centering
\caption{Maximum Detection Range at Apolune ($v_\mathrm{rel}=10$~m/s,
Swerling~I, $K=16$ CPIs)}
\label{tab:rmax}
\begin{tabular}{lcccc}
\hline
\textbf{Target} & \textbf{Size} & \textbf{RCS} & \textbf{$R_\mathrm{max}$} & 
\textbf{Warning} \\
 & [m] & [dBsm] & [km] & [min] \\
\hline
Fragment & 0.3 & $-11.5$ & 53 & 88 \\
Object & 0.5 & $-7.1$ & 69 & 115 \\
Rocket body & 1.0 & $-1.0$ & 97 & 162 \\
Upper stage & 2.0 & $+5.0$ & 137 & 228 \\
Spent stage & 5.0 & $+12.9$ & 217 & 362 \\
\hline
\end{tabular}
\end{table}

\begin{proposition}[CPI Configuration]
\label{prop:cpi}
Given a total observation window $T_\mathrm{obs}$ during a relay 
session, the CPI duration that maximizes detection SNR is
\begin{equation}
T_\mathrm{CPI}^*(\theta) = 
\min\!\left(\frac{c}{2B \cdot v_\mathrm{rel}(\theta)}, \; 
T_\mathrm{obs}\right)
\label{eq:Tcpi_opt}
\end{equation}
i.e., each CPI is as long as possible before range migration 
exceeds one resolution cell $\Delta R = c/(2B)$.
The corresponding optimal number of CPIs is 
$K^*(\theta) = \lfloor T_\mathrm{obs}/T_\mathrm{CPI}^*(\theta) \rfloor$.
\end{proposition}

\begin{IEEEproof}
Within the range-migration-free regime 
($v_\mathrm{rel} T_\mathrm{CPI} \leq c/(2B)$), 
$\gamma_\mathrm{det} \propto M K^\kappa \propto 
K^{\kappa-1} = K^{-0.15}$, so fewer, longer CPIs are 
preferred. Beyond the migration limit, coherent gain 
saturates, yielding \eqref{eq:Tcpi_opt}.
\end{IEEEproof}

At apolune ($v_\mathrm{rel} = 10$~m/s), the range migration limit
gives $T_\mathrm{CPI}^* = 0.15$~s, allowing long coherent 
integration. At perilune ($v_\mathrm{rel} = 500$~m/s), 
$T_\mathrm{CPI}^* = 3.0$~ms, requiring many short CPIs and 
extensive non-coherent integration.
This velocity-dependent CPI structure motivates the adaptive 
processing design in Section~IV.

An additional cap arises from finite beam dwell. For a fixed 
beam ($\theta_\mathrm{bw} = 0.62^\circ$), the central diametric 
crossing gives an upper-bound dwell 
$T_\mathrm{dwell} = 2R\tan(\theta_\mathrm{bw}/2)/v_\mathrm{rel}$ 
with $K \leq T_\mathrm{dwell}/T_\mathrm{CPI}$; off-center 
trajectories yield shorter dwells (the average random-chord 
benchmark is $\pi/4 \approx 0.79\times$ the diametric value), 
while active gimbal steering extends dwell by orders of 
magnitude. We retain the diametric value as an optimistic 
fixed-beam reference; mission-stage analysis with realistic 
trajectory geometry will tighten the resulting $K$ cap. Since 
$R_\mathrm{max}$ and $T_\mathrm{dwell}$ are mutually 
dependent, the joint $(K, R_\mathrm{max})$ is resolved by 
fixed-point iteration (typically $<\!10$ steps).

\begin{figure*}[!t]
\centering
\includegraphics[width=0.85\textwidth]{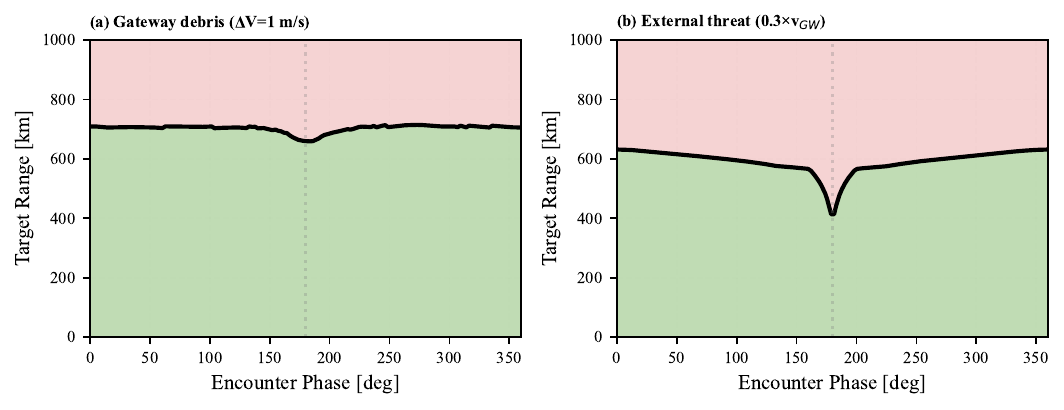}
\caption{Detection feasibility (green: 
$\mathrm{SNR}\geq\gamma_\mathrm{th}$; red: sub-threshold) 
over the encounter-phase / target-range plane for 
(a)~Gateway debris ($\Delta v=1$~m/s, CR3BP) and 
(b)~external threats ($v_\mathrm{rel}=0.3 v_\mathrm{GW}(\theta)$). 
Black contour: $R_\mathrm{max}(\theta)$. 
$K=K^*(\theta)$, $\rho^*=0.6$.}
\label{fig:coverage}
\end{figure*}

\section{Velocity-Adaptive Processing and Resource Allocation}
\label{sec:design}
\begin{figure}[!t]
\centering
\includegraphics[width=\columnwidth]{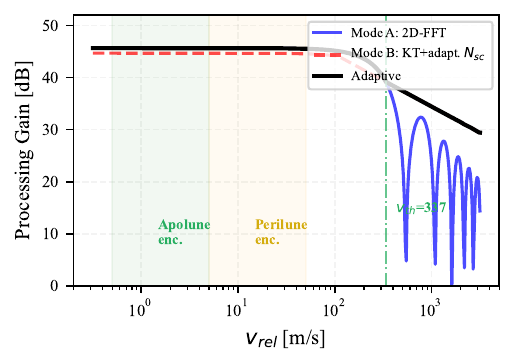}
\caption{Coherent processing gain vs.\ relative velocity 
for Mode~A (2D-FFT), Mode~B (KT + adaptive $N_\mathrm{sc}$), 
and the adaptive envelope. Mode-switching threshold 
$v_\mathrm{th}=337$~m/s. Shaded bands: CR3BP apolune and 
perilune encounter regimes; all Gateway debris falls in 
the flat-gain region.}
\label{fig:processing_gain}
\end{figure}

The CRB analysis in Section~\ref{sec:crb} reveals that sensing 
performance degrades near perilune due to ICI from high relative 
velocity.
As shown in Section~\ref{sec:results}, Gateway's own operational 
debris operates in the low-velocity regime 
($v_\mathrm{rel} < 50$~m/s) where standard processing suffices.
The adaptive design developed below targets the higher-velocity 
regime ($v_\mathrm{rel} > 100$~m/s) anticipated from 
non-cooperative objects as cislunar traffic grows 
\cite{Barakat2024}.
This section addresses the degradation through two mechanisms:
velocity-adaptive waveform processing 
(Section~\ref{sec:mode_ab}--\ref{sec:switching}) and 
orbit-adaptive resource allocation 
(Section~\ref{sec:allocation}).

\subsection{Dual-Mode Processing}
\label{sec:mode_ab}

\emph{Mode~A} (standard 2D-FFT \cite{Braun2014}) yields 
$G_A = NM \cdot \mathcal{G}(v_\mathrm{rel})$; at apolune 
($v_\mathrm{rel} \approx 1$--$10$~m/s), 
$\mathcal{G} \approx \eta_\mathrm{impl}$ and the full gain 
is available.
\emph{Mode~B} addresses two degradation mechanisms at high 
$v_\mathrm{rel}$: range migration (target traverses multiple 
range bins during one CPI) and ICI (signal energy 
redistributed across adjacent subcarriers by the factor 
$\mathrm{sinc}^2(f_d/\Delta f)$).

Range migration is corrected by the Keystone Transform (KT) 
\cite{Perry1999} ($t' = t f_c/(f_c+f)$, $\sim$1~dB loss).
ICI is mitigated by reducing subcarriers to 
$N_\mathrm{sc}^{(B)}$, widening $\Delta f^{(B)} = 
B/N_\mathrm{sc}^{(B)}$ and reducing $f_d/\Delta f^{(B)}$.
The optimal subcarrier count satisfies the ICI tolerance 
condition:
\begin{equation}
N_\mathrm{sc}^{(B)}(v_\mathrm{rel}) = 
\max\!\left(1,\;\min\!\left(\left\lfloor 
\frac{0.2\,B}{2 v_\mathrm{rel} f_c / c} 
\right\rfloor, \; N\right)\right)
\label{eq:Nsc_opt}
\end{equation}
which ensures $f_d/\Delta f^{(B)} < 0.2$, corresponding to 
less than 0.5~dB ICI loss.
The tradeoff is reduced range resolution and coherent range 
compression gain: $N_\mathrm{sc}^{(B)} < N$.

The Mode~B processing gain is
\begin{equation}
G_B = N_\mathrm{sc}^{(B)} \cdot M \cdot 
\mathcal{G}^{(B)}(v_\mathrm{rel}) \cdot \eta_\mathrm{KT}
\cdot \eta_\mathrm{impl}
\label{eq:G_B}
\end{equation}
where $\mathcal{G}^{(B)} = 
\mathrm{sinc}^2(f_d/\Delta f^{(B)})$ uses the wider subcarrier 
spacing, and $\eta_\mathrm{KT} \approx 0.8$ ($-1$~dB) accounts 
for KT interpolation loss.

\subsection{Optimal Mode Switching}
\label{sec:switching}

The adaptive processor selects the mode that yields the higher 
processing gain:
\begin{equation}
G^*\!(v_\mathrm{rel}) = \max\{G_A, \, G_B\}.
\label{eq:G_star}
\end{equation}
Since $G_A$ decreases monotonically with $v_\mathrm{rel}$ 
(due to $\mathcal{G} \to 0$) while $G_B$ decreases more slowly
(due to reduced $N_\mathrm{sc}^{(B)}$ compensating the ICI loss),
a unique crossover velocity $v_\mathrm{th}$ exists where 
$G_A(v_\mathrm{th}) = G_B(v_\mathrm{th})$.
For the HLCS parameters in Table~\ref{tab:params}, numerical 
evaluation yields $v_\mathrm{th} \approx 337$~m/s.

The gain recovery from Mode~B is moderate: at 
$v_\mathrm{rel} = 500$~m/s, Mode~B provides 37.5~dB processing 
gain versus 24.1~dB for Mode~A, recovering approximately 3~dB 
in detection range.
At $v_\mathrm{rel} > 2{,}000$~m/s, the required reduction in 
$N_\mathrm{sc}^{(B)}$ is so severe that Mode~B offers no 
improvement over Mode~A.
This result confirms that the NRHO velocity structure, rather 
than signal processing limitations, fundamentally determines the 
sensing capability near perilune.

\begin{remark}[Slow-Time Doppler Ambiguity]
\label{rem:doppler_ambig}
With OFDM symbol duration $T_\mathrm{sym}\!\approx\! 11.5$~$\mu$s, 
the slow-time Doppler sampling rate $1/T_\mathrm{sym}\!\approx\! 
86.8$~kHz yields an unambiguous Doppler interval 
$|f_d|<43.4$~kHz, corresponding to an unambiguous radial 
velocity $|v_\mathrm{rel}|<c/(4f_c T_\mathrm{sym})\!\approx\! 
241$~m/s. The Mode~A/Mode~B analyses are ambiguity-free in 
this regime, which fully covers Gateway operational debris 
($v_\mathrm{rel}<50$~m/s). The external-threat regime up to 
$500$~m/s exceeds the slow-time Nyquist limit; de-aliasing 
in this regime requires standard auxiliary techniques such 
as multi-CPI phase unwrapping, staggered symbol timing, or 
prior-bounded track-aided ambiguity resolution. Detailed 
implementation is deferred to mission-stage processing 
design and does not affect the SNR-limited range conclusions 
Section~\ref{sec:results}.
\end{remark}

\subsection{Orbit-Adaptive Sensing-Communication Allocation}
\label{sec:allocation}

We formulate the resource allocation problem as follows.
During each relay session of duration $T_\mathrm{obs}$, the 
receiver alternates between communication reception and sensing 
echo capture in a TDD manner (Remark~\ref{rem:sic}).
Let $\rho(\theta) \in [0,1]$ denote the fraction of the session 
allocated to sensing at orbital phase $\theta$.
The design variables at each phase are $\rho(\theta)$, 
$T_\mathrm{CPI}(\theta)$, and $N_\mathrm{sc}(\theta)$.

\begin{proposition}[Sequential Sensing-Communication Allocation]
\label{prop:alloc}
The max-min detection range problem
\begin{equation}
\max_{\rho(\theta),\, T_\mathrm{CPI}(\theta),\, 
N_\mathrm{sc}(\theta)} \;\; \min_\theta \; 
R_\mathrm{max}(\theta)
\label{eq:opt}
\end{equation}
subject to communication throughput 
$R_\mathrm{comm}(\theta) = (1-\rho(\theta)) R_\mathrm{full} 
\geq R_\mathrm{min}$ admits a sequential monotone solution:
\begin{align}
\rho^* &= 1 - R_\mathrm{min}/R_\mathrm{full}
\label{eq:rho_opt} \\
T_\mathrm{CPI}^*(\theta) &= 
\min\!\left(\frac{c}{2B v_\mathrm{rel}(\theta)},\; 
T_\mathrm{obs}\right)
\label{eq:Tcpi_alloc} \\
N_\mathrm{sc}^*(\theta) &= N_\mathrm{sc}^{(B)}(v_\mathrm{rel}(\theta))
\label{eq:Nsc_alloc} \\
K^*(\theta) &= 
\left\lfloor \frac{\rho^* T_\mathrm{obs}}{T_\mathrm{CPI}^*(\theta)} 
\right\rfloor
\label{eq:K_alloc}
\end{align}
where $R_\mathrm{full}$ is the link throughput at full 
communication allocation and $R_\mathrm{min}$ is the minimum 
required throughput.
\end{proposition}

\begin{IEEEproof}
We optimize each variable sequentially by monotonicity.
First, $N_\mathrm{sc}$ affects only the per-CPI processing 
efficiency $\mathcal{G}$ and is optimized independently by 
\eqref{eq:Nsc_opt}.
Second, for any fixed $\rho$, the CPI duration 
$T_\mathrm{CPI}$ determines both the coherent gain $M$ and 
the number of CPIs $K = \lfloor \rho T_\mathrm{obs}/T_\mathrm{CPI} \rfloor$.
Since $R_\mathrm{max} \propto (M \cdot K^\kappa)^{1/4}$ and 
$M \cdot K^\kappa \propto T_\mathrm{CPI}^{1-\kappa}$ within 
the range-migration-free regime 
(Proposition~\ref{prop:cpi}), $R_\mathrm{max}$ is monotonically 
increasing in $T_\mathrm{CPI}$ for $\kappa < 1$.
The maximum is therefore attained at the range migration limit 
\eqref{eq:Tcpi_opt}, yielding \eqref{eq:Tcpi_alloc}.
Third, with $T_\mathrm{CPI}^*$ fixed, $R_\mathrm{max}$ is 
monotonically increasing in $\rho$ through $K$.
The communication constraint 
$\rho \leq 1 - R_\mathrm{min}/R_\mathrm{full}$ binds at the 
optimum, giving \eqref{eq:rho_opt} and \eqref{eq:K_alloc}.
\end{IEEEproof}

The throughput model $R_\mathrm{comm} = (1-\rho)R_\mathrm{full}$ 
is a first-order TDD scheduling abstraction; adaptive 
modulation/coding, queueing, and latency effects are deferred 
to a communication-focused extension.

For the HLCS, $R_\mathrm{full}(\theta)$ varies from 
$\sim$104 Mbps at apolune to $\sim$116 Mbps at perilune. 
Substituting the worst-case (apolune) value into 
Proposition~\ref{prop:alloc} with $R_\mathrm{min} = 40$~Mbps 
yields the constant $\rho^* = 0.6$ that satisfies the 
throughput constraint at every orbital phase, serving as a 
conservative baseline. This constant policy is uniform across 
the orbit and therefore cannot exploit the orbit-phase 
variation of either the relay link quality 
$R_\mathrm{full}(\theta)$ or the sensing requirement 
$R_\mathrm{target}(\theta)$, which depends on local debris 
encounter density. The following proposition relaxes the 
constant-$\rho$ restriction to exploit both dependencies.

\begin{proposition}[Orbit-Phase-Adaptive Allocation]
\label{prop:adaptive}
Let the relay throughput $R_\mathrm{full}(\theta)$ and minimum 
detection range $R_\mathrm{target}(\theta)$ both depend on orbital 
phase.
The orbit-averaged throughput maximization
\begin{equation}
\max_{\rho(\theta)} \;
\int_0^{2\pi} \!\bigl(1-\rho(\theta)\bigr)\,
R_\mathrm{full}(\theta)\,w(\theta)\,d\theta
\label{eq:opt_adaptive}
\end{equation}
subject to $R_\mathrm{max}\bigl(\theta,\rho(\theta)\bigr) 
\geq R_\mathrm{target}(\theta)$ for all $\theta$ and 
$0 \leq \rho(\theta) \leq 1$, where 
$w(\theta)$ is the dwell-time weight extracted from the CR3BP 
reference trajectory (approximately $\propto r(\theta)^2$),
admits the following pointwise solution under continuous 
relaxation of $K$ and non-binding beam dwell
\begin{equation}
\rho^*\!(\theta) = 
\frac{T_\mathrm{CPI}^*(\theta)}{T_\mathrm{obs}} \cdot
\left(\frac{R_\mathrm{target}(\theta)^4}
{C(\theta)}\right)^{\!\!1/\kappa}
\label{eq:rho_adaptive}
\end{equation}
where $C(\theta) = \mathrm{EIRP} \cdot G_r \lambda^2 \sigma 
\cdot N_\mathrm{sc}^*\!(\theta) M^*\!(\theta) 
\mathcal{G}^*\!(\theta) /
\bigl[(4\pi)^3 k_B T_\mathrm{sys} B \gamma_\mathrm{th}\bigr]$
is the single-CPI sensing figure of merit at phase~$\theta$.
\end{proposition}

\begin{IEEEproof}
Since the objective \eqref{eq:opt_adaptive} is separable across 
$\theta$ (no coupling between different orbital phases) and 
$R_\mathrm{max}$ at each phase depends only on the local 
$\rho(\theta)$, the problem decomposes into independent 
pointwise optimizations.
At each $\theta$, the throughput 
$(1-\rho)\,R_\mathrm{full}(\theta)$ is monotonically decreasing 
in $\rho$, while $R_\mathrm{max}(\theta,\rho)$ is monotonically 
increasing in $\rho$ through $K = \rho T_\mathrm{obs}/T_\mathrm{CPI}^*$.
The unique optimum therefore occurs where the detection 
constraint binds: 
$R_\mathrm{max}(\theta,\rho^*) = R_\mathrm{target}(\theta)$.
Substituting $R_\mathrm{max}^4 = C(\theta)\cdot K^\kappa$ with 
$K = \rho T_\mathrm{obs}/T_\mathrm{CPI}^*$ and solving for 
$\rho$ yields \eqref{eq:rho_adaptive}.
\end{IEEEproof}

The closed form \eqref{eq:rho_adaptive} assumes continuous 
relaxation of $K$ and non-binding beam dwell 
(Section~\ref{sec:crb}).
Numerical evaluation of the exact discrete implementation 
(integer $K$, beam-dwell-constrained fixed-point iteration) 
shows that the orbit-averaged throughput is within 2\% of 
the relaxed solution, confirming that the continuous 
relaxation is tight for the HLCS parameters in 
Table~\ref{tab:params}.

Since $\rho^*(\theta) \propto T_\mathrm{CPI}^*(\theta)$ 
from \eqref{eq:rho_adaptive} and $T_\mathrm{CPI}^*$ 
decreases with velocity (Proposition~\ref{prop:cpi}), 
both the perilune regime (many short CPIs) and the apolune 
regime (few long CPIs) require only modest $\rho$, yielding 
$\bar{\rho} \approx 0.19$ versus $\rho^* = 0.6$ under the 
constant allocation (Section~\ref{sec:adaptive_results}).

In practice, near-range sensing additionally exploits TDD guard 
periods that would otherwise be idle (Remark~\ref{rem:sic}), 
though the primary sensing gain comes from the dedicated 
sensing slots.

The computational overhead is modest: mode selection is a 
single sinc evaluation, KT adds $O(NM\log M)$, and the 
$K$--$R_\mathrm{max}$ solver converges in fewer than 10 
iterations---all feasible on space-qualified FPGAs.

\section{Debris Encounter Model and Simulation Setup}
\label{sec:setup}

\begin{figure*}[!t]
\centering
\includegraphics[width=0.96\textwidth]{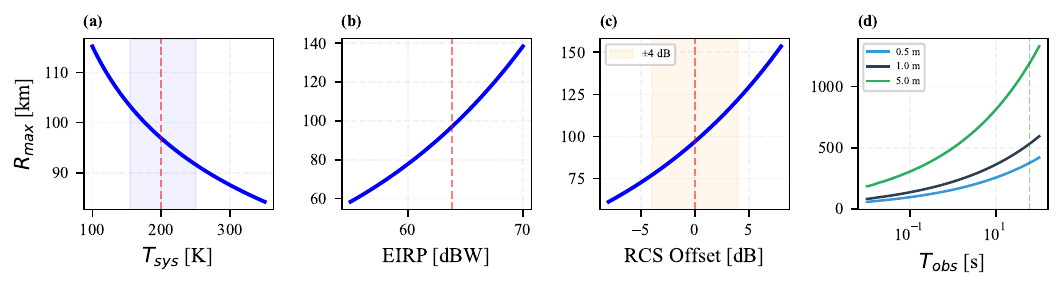}
\caption{Sensitivity of snapshot detection range to (a) system noise temperature $T_\mathrm{sys}$, (b) EIRP, (c) RCS uncertainty ($\pm 4$~dB from NASA SEM), and (d) observation time $T_\mathrm{obs}$ for three target sizes at $v_\mathrm{rel} = 10$~m/s with beam dwell constraint. Red dashed lines indicate nominal values.}
\label{fig:sensitivity}
\end{figure*}

\subsection{CR3BP-Based Recontact Trajectory Model}
\label{sec:cr3bp}

The NRHO's marginal instability ($\nu \approx 1.3$ 
\cite{Zimovan2017}) keeps separated objects ($\Delta v = 
1$--5~m/s) in the vicinity for 30--100~days, with recontact 
every $\sim$3.3~days near perilune \cite{Davis2019}.
Both trajectories are propagated in the Earth--Moon CR3BP 
($\mu = 0.01215$) \cite{Szebehely1967} using a 
differentially corrected 9:2 NRHO \cite{Williams2017}, 
integrated (DOP853, $10^{-12}$ tolerance) for 30--45~days.
Solar radiation pressure is neglected: for compact debris 
(area-to-mass ratio $< 0.01$~m$^2$/kg), SRP perturbation is 
below 1\% of the CR3BP acceleration at lunar distance.
For high-AMR fragments (e.g., multilayer insulation), SRP can 
significantly alter recontact timing and is deferred to 
future work.
Davis et al.~\cite{Davis2019} validated CR3BP recontact 
predictions against full ephemeris for the first several 
revolutions.

A campaign over 24 separation phases, 3 directions (VNB frame), 
and $\Delta v\in\{1,5\}$~m/s shows that perilune encounters 
amplify $v_\mathrm{rel}$ by up to $50\times\Delta v$ while 
apolune encounters yield $v_\mathrm{rel}\approx\Delta v$; 
normal-direction separation is the most persistent, and 
$\Delta v\geq 5$~m/s in the velocity direction eliminates 
recontact. The resulting $v_\mathrm{rel}(\theta)$ profile feeds 
the detection analysis in Section~\ref{sec:results}.

\subsection{Simulation Parameters}

The simulation parameters are summarized in 
Table~\ref{tab:sim_params}.
We define three debris encounter scenarios based on the CR3BP 
recontact data:
\begin{enumerate}
\item \textbf{Recent separation} ($t_\mathrm{sep} < 3$~days):
$v_\mathrm{rel} = 1$--$5$~m/s. The debris has not yet completed 
its first half-revolution and remains close to the Gateway with 
low relative velocity.
\item \textbf{First recontact} 
($t_\mathrm{sep} \approx 7$--$10$~days): 
$v_\mathrm{rel} = 5$--$50$~m/s.
The debris returns to the Gateway vicinity after its first 
perilune passage, with amplified relative velocity.
\item \textbf{Long-term revisit} 
($t_\mathrm{sep} \approx 20$--$30$~days): 
$v_\mathrm{rel} = 1$--$30$~m/s.
Multiple recontact events with variable velocity, reflecting the 
complex interplay between unstable-manifold divergence and 
periodic orbit structure.
\end{enumerate}

\begin{table}[!t]
\centering
\caption{Simulation Parameters}
\label{tab:sim_params}
\begin{tabular}{llll}
\hline
\textbf{Parameter} & \textbf{Value} & \textbf{Source} & \textbf{Type} \\
\hline
Carrier frequency $f_c$ & 27 GHz & \cite{Dellandrea2022} & Doc. \\
Wavelength $\lambda$ & 11.1 mm & & Derived \\
Bandwidth $B$ & 100 MHz & \cite{Johnson2021} & Doc. \\
Antenna diameter $D$ & 1.25 m & \cite{Dellandrea2022} & Doc. \\
Antenna gain $G_t = G_r$ & 48.4 dBi & $\eta = 0.55$ & Derived \\
TX power $P_t$ & 35 W (15.4 dBW) & \cite{Dellandrea2022} & Doc. \\
EIRP & 63.8 dBW & & Derived \\
System temp.\ $T_\mathrm{sys}$ & 200 K & \cite{KaBandLNA2001} & Doc. \\
Subcarriers $N$ & 1024 & & Assumed \\
Symbols/CPI $M$ & 64 & & Assumed \\
Subcarrier spacing $\Delta f$ & 97.7 kHz & & Derived \\
CPI duration $T_\mathrm{CPI}$ & 0.74 ms & & Derived \\
\hline
\multicolumn{4}{l}{\footnotesize Doc.\ = mission-documented; 
Derived = calculated from documented values;} \\
\multicolumn{4}{l}{\footnotesize Assumed = author-selected 
for sensing analysis.} \\
\end{tabular}
\end{table}

All sensing performance metrics (CRB, $R_\mathrm{max}$, warning 
time) are evaluated for a target with diameter 1~m 
($\sigma = \pi/4 \approx -1$~dBsm, optical regime at 27~GHz; 
see Section~\ref{sec:system_model}).
Results for other target sizes are obtained by scaling 
$R_\mathrm{max} \propto \sigma^{1/4}$.
The self-consistent $K$--$R_\mathrm{max}$ relationship, 
accounting for beam dwell time, is resolved iteratively as 
described in Section~\ref{sec:crb}.

\section{Numerical Results}
\label{sec:results}

This section evaluates the proposed framework using the CR3BP 
debris encounter model and HLCS parameters from 
Sections~\ref{sec:setup} and \ref{sec:system_model}.
All results use the primary operating point of 
$T_\mathrm{obs} = 60$~s with $\rho^* = 0.6$ unless stated 
otherwise.

\subsection{Cislunar Environmental Advantage}

Fig.~\ref{fig:snr_advantage}(a) plots the detection SNR versus target range at apolune 
for four debris diameters (0.5, 1, 2, 5~m) and a Swerling~I 
threshold of 21.1~dB.
The 1~m target crosses the detection threshold at 97~km 
(snapshot, $K = 16$), while the 5~m spent stage remains 
detectable beyond 200~km.
The ground-based Ka-band reference curve for a 1~m target, 
shifted down by 35.6~dB to account for atmospheric absorption, 
ionospheric scintillation, ground clutter, and the elevated 
noise temperature typical of terrestrial Ka-band SSA receivers
(Fig.~\ref{fig:snr_advantage}(b)), crosses the threshold at 
only 12~km, explaining why ground-based cislunar SSA is not 
operationally feasible \cite{NASARadar2023,Barakat2024}. 
The 36~dB advantage translates to a $7.8\times$ improvement 
in detection range for an in-situ Gateway sensor under the 
$R^4$ path loss law.

\subsection{Debris Recontact Velocity Profile}

Fig.~\ref{fig:vrel_bands} shows that apolune encounters yield 
$v_\mathrm{rel} \approx \Delta v$, while perilune encounters 
amplify it by up to $50\times$ (49.5~m/s for $\Delta v = 1$~m/s, 
86.7~m/s for 5~m/s), consistent with 
\cite{Davis2019,Scheuerle2023}.

\subsection{Detection Range and Warning Time}

Fig.~\ref{fig:rmax_warning}(a) compares the maximum detection range along the NRHO 
for four scenarios.
The CR3BP recontact curve for a 1~m target (blue, with shaded 
band spanning median to 95th-percentile $v_\mathrm{rel}$) remains 
between 670 and 710~km across the entire orbit, with only a 
minor dip near perilune.
This confirms that Gateway's own operational debris, with 
$v_\mathrm{rel} \leq 50$~m/s, operates entirely within the 
ICI-free regime.
Three parametric external-threat curves 
($v_\mathrm{rel} = 0.3 \times v_\mathrm{GW}(\theta)$, 
spanning plausible non-cooperative encounter kinematics) for 
target diameters 0.5, 1.0, and 2.0~m show orbit-phase-dependent 
detection range: the 1~m case drops from 630~km at apolune 
to 400~km at perilune (35\% reduction), while the 2~m target 
remains above 630~km throughout.
Mission-specific threat trajectories require detailed conjunction 
analysis beyond the scope of this work.
Fig.~\ref{fig:rmax_warning}(b) translates these ranges into 
warning time.
Gateway debris warning exceeds $10^4$~minutes at apolune and 
2{,}000~minutes at perilune, far above the 3-hour COLA threshold.
The 1~m external threat drops to 12~minutes at perilune but 
remains above the 5-minute emergency threshold.
The snapshot baseline ($K = 16$, dotted) falls below 1~minute 
near perilune, confirming that extended observation is essential.

\subsection{Processing Gain and Mode Selection}

Fig.~\ref{fig:processing_gain} shows that the standard 2D-FFT processor (Mode~A) 
provides the full 45.7~dB coherent gain for 
$v_\mathrm{rel} < 100$~m/s, encompassing all CR3BP recontact 
encounters.
Mode~B (KT with adaptive $N_\mathrm{sc}$) incurs 1~dB overhead 
from KT interpolation but degrades more gracefully at high 
velocity, crossing Mode~A at $v_\mathrm{th} = 337$~m/s.
The adaptive envelope recovers 13~dB over Mode~A at 
$v_\mathrm{rel} = 500$~m/s.
For Gateway debris, adaptation provides no benefit; for external 
threats with $v_\mathrm{rel} > 300$~m/s, Mode~B recovers 
approximately 3~dB in detection range.

The detection range scales as $R_\mathrm{max} \propto d^{1/2}$ 
(from $\sigma \propto d^2$ and $R_\mathrm{max} \propto 
\sigma^{1/4}$); Table~\ref{tab:rmax} gives values for five 
representative target sizes.
The ICI penalty at $v_\mathrm{rel} = 500$~m/s reduces the 
snapshot detection range by approximately $3.5\times$ 
relative to the apolune baseline ($97/28\!\approx\!3.46$, 
consistent with the $21.6$~dB ICI loss).

\subsection{Sensitivity Analysis}

Fig.~\ref{fig:sensitivity} ($K=16$ snapshot) shows 
$R_\mathrm{max}$ scales as $T_\mathrm{sys}^{-1/4}$, 
$\mathrm{EIRP}^{1/4}$, and incurs up to 26\% variation from 
$\pm 4$~dB RCS uncertainty \cite{XuKennedy2019}; extending 
$T_\mathrm{obs}$ from snapshot to 60~s lifts the 1~m range 
from 97 to $\sim$550~km, with growth saturating under 
$K^{0.21}$ scaling and beam dwell.

\subsection{Sensing-Communication Tradeoff}

At $\rho^* = 0.6$, the system achieves 
$R_\mathrm{max} \approx 650$~km at $v_\mathrm{rel}=10$~m/s and 
$\sim$400~km at 500~m/s while maintaining 40~Mbps throughput, 
with the constraint binding at the optimum of 
Proposition~\ref{prop:alloc}. A time-sharing policy at the 
same $\bar{\rho}=0.19$ achieves identical throughput but 
concentrates sensing into $\sim$19\% of the orbit, creating 
coverage gaps; the adaptive allocation in 
Proposition~\ref{prop:adaptive} instead distributes sensing 
across all phases, keeping Gateway debris and external threats 
above the 100~km keep-out sphere \cite{Davis2019} throughout.

\begin{figure*}[!t]
\centering
\includegraphics[width=0.92\textwidth]{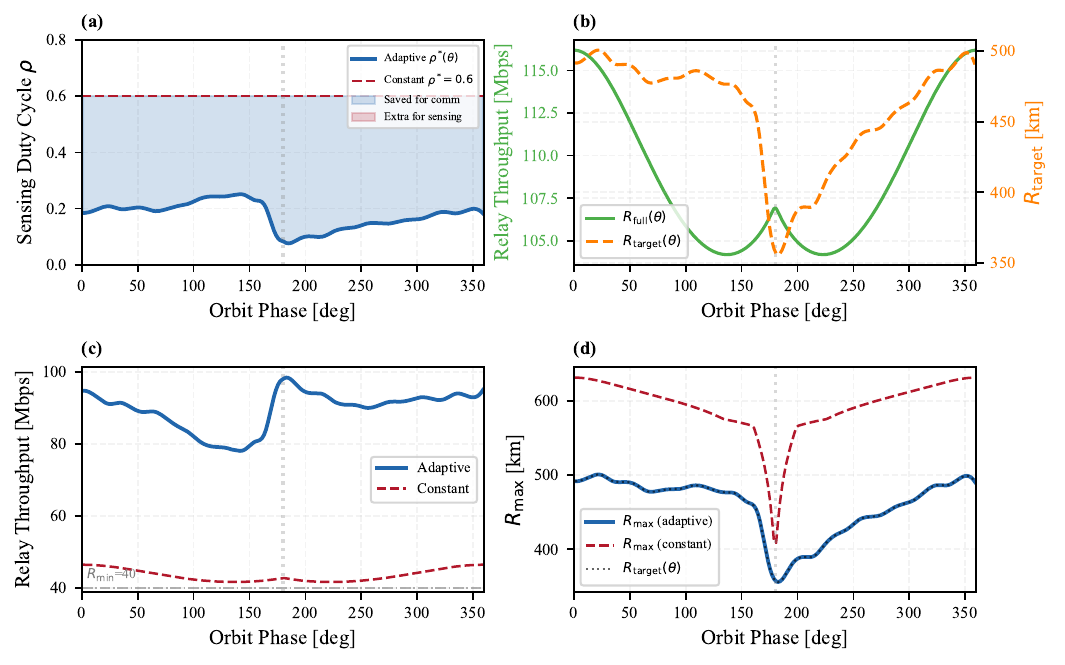}
\caption{Orbit-phase-adaptive allocation 
(Proposition~\ref{prop:adaptive}) vs.\ constant 
$\rho^*=0.6$ baseline. (a)~Adaptive duty cycle 
$\rho^*(\theta)\in[0.08,0.25]$. 
(b)~$R_\mathrm{full}(\theta)$ and risk-weighted 
$R_\mathrm{target}(\theta)$. (c)~Adaptive throughput: 
90~Mbps orbit-averaged vs.\ 44~Mbps. (d)~Both meet 
$R_\mathrm{target}(\theta)$; adaptive uses 69\% less 
sensing time.}
\label{fig:adaptive_rho}
\end{figure*}

\subsection{Orbit-Phase-Adaptive Allocation}
\label{sec:adaptive_results}

To apply Proposition~\ref{prop:adaptive}, the detection 
coverage requirement $R_\mathrm{target}(\theta)$ is 
constructed from the CR3BP recontact encounter density 
$n(\theta)$ (Section~\ref{sec:cr3bp}):
\begin{equation}
R_\mathrm{target}(\theta) = R_\mathrm{base} + 
R_\mathrm{risk}\,\frac{n(\theta)}{n_\mathrm{max}}
\label{eq:R_target}
\end{equation}
where $R_\mathrm{base} = 200$~km is the minimum keep-out 
range and $R_\mathrm{risk} = 300$~km scales the additional 
coverage with local encounter density.
This heuristic coverage profile is a design-policy parameter, 
not a mission-safety-derived constraint; a rigorous derivation 
from conjunction probability and miss-distance distributions 
is deferred to operational mission analysis.
This yields $R_\mathrm{target}$ between 360~km (perilune, 
lowest encounter density) and 500~km (near apolune, highest 
encounter density).
The orbit-averaged duty cycle remains below 0.30 for 
$R_\mathrm{base} \in [150, 250]$~km and 
$R_\mathrm{risk} \in [200, 400]$~km, confirming structural 
robustness of the adaptive allocation.
In all tested combinations, the qualitative conclusion 
holds: the adaptive policy requires less than half the 
duty cycle of the constant allocation while maintaining 
detection coverage at every orbital phase.

Fig.~\ref{fig:adaptive_rho} compares the orbit-phase-adaptive 
allocation from Proposition~\ref{prop:adaptive} with the 
constant $\rho^* = 0.6$ baseline from 
Proposition~\ref{prop:alloc}.
Panel~(a) shows $\rho^*(\theta)$ ranging from 0.08 at 
perilune to 0.25 at intermediate phases 
($\bar{\rho} = 0.19$, less than one-third of the constant 
$\rho = 0.6$), because the constant baseline massively 
over-allocates at apolune where detection is straightforward.
Panel~(b) illustrates the anti-correlated demand structure: 
$R_\mathrm{full}(\theta)$ peaks near perilune while 
$R_\mathrm{target}(\theta)$ peaks at intermediate phases.
Panel~(c) shows that the adaptive policy achieves 90~Mbps 
orbit-averaged throughput (vs.\ 44~Mbps at $\rho = 0.6$).
The best constant $\rho$ satisfying the same constraints at 
all orbital phases is $\rho = 0.26$ (81.9~Mbps); the 
adaptive policy contributes an additional 8.1~Mbps (10\% 
relative to the best constant baseline, or 18\% of the 
46~Mbps total gain over the conservative $\rho = 0.6$ 
baseline) by tightening coverage at every orbital phase 
(Panel~(d)).

The orbit-phase structure of the NRHO thus creates a 
scheduling opportunity unavailable in circular orbits.

\begin{figure*}[!t]
\centering
\includegraphics[width=0.92\textwidth]{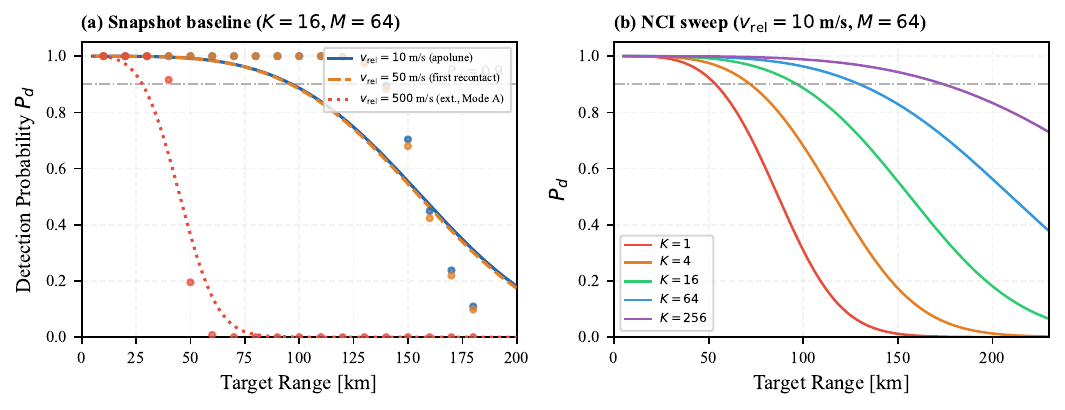}
\caption{Snapshot-baseline detection probability at 
$(N,M,K)=(1024,64,16)$. (a)~$P_d$ vs.\ range at three 
relative velocities; markers: $2\times 10^4$ Monte Carlo 
trials, lines: $P_d=P_{fa}^{1/(1+\bar{\gamma})}$. The 
$v_\mathrm{rel}=500$~m/s case collapses to 
$R_\mathrm{max,det}\approx 28$~km under Mode~A. 
(b)~$K$-sweep at $v_\mathrm{rel}=10$~m/s, $M=64$: 
$R_\mathrm{max,det}$ grows from $\approx 54$~km ($K=1$) 
to $\approx 174$~km ($K=256$).}
\label{fig:mc_pd}
\end{figure*}

\subsection{Monte Carlo Detection Validation}
\label{sec:mc_pd}

Figure~\ref{fig:mc_pd} characterizes the snapshot-baseline
detection performance, i.e., under the fixed configuration
$(N,M,K)=(1024,64,16)$ used in Table~\ref{tab:rmax} before
session-level optimization is applied. This regime isolates the
behavior of Theorem~\ref{thm:rmax} at a single CPI, prior to the
per-phase coherent/non-coherent rebalancing of
Proposition~\ref{prop:cpi}.

Panel~(a) shows $P_d(R)$ for three representative relative
velocities. At $v_\mathrm{rel}=10$~m/s (apolune near-stationary
phase), the theoretical curve matches Monte Carlo simulation
within $0.4$~dB and the $P_d=0.9$ operating point is reached at
$R_\mathrm{max,det}\approx 97$~km, consistent with
Table~\ref{tab:rmax}. At $v_\mathrm{rel}=50$~m/s (first
post-apolune recontact), Doppler-induced ICI is still negligible
and $R_\mathrm{max,det}$ is essentially unchanged. At
$v_\mathrm{rel}=500$~m/s, representative of an external
high-energy threat under Mode~A (standard 2D-FFT) processing,
ICI becomes dominant and $R_\mathrm{max,det}$ collapses to
$\approx 28$~km. This collapse is the quantitative motivation
for the Mode~B (Keystone Transform with adaptive $N_\mathrm{sc}$)
mitigation introduced in Section~\ref{sec:design}; the
corresponding recovery is visible in Fig.~\ref{fig:processing_gain}.

Panel~(b) isolates the effect of the NCI depth $K$ at fixed
$M=64$ and $v_\mathrm{rel}=10$~m/s. As predicted by the
$K^{\kappa}$ scaling of Theorem~\ref{thm:rmax},
$R_\mathrm{max,det}$ grows from approximately $54$~km at $K{=}1$
(single CPI) to $174$~km at $K{=}256$, a $3.2\times$ span. This
confirms the tightness of the closed-form expression across more
than two decades of NCI depth under the snapshot parameters.

The snapshot-baseline ranges in Fig.~\ref{fig:mc_pd} are
deliberately conservative: they fix $M=64$ rather than allowing
$M$ to adapt to the instantaneous relative velocity. When the
full session-level allocation of Proposition~\ref{prop:alloc} is
applied---i.e., $T_\mathrm{CPI}^{*}(\theta)$ and the resulting
$(N_\mathrm{sc},M,K)$ allocation are re-optimized per orbital
phase---the achievable detection ranges extend significantly
further, reaching $\sim 700$~km for Gateway co-moving debris and
$\sim 400$~km for external cislunar threats at perilune, as
shown in Figs.~\ref{fig:rmax_warning} and \ref{fig:outage}.

The close theory--MC agreement across all $K$ implicitly 
validates the fixed $\kappa=0.85$ approximation, which incurs 
$<\!6\%$ variation in $R_\mathrm{max}$ across the relevant 
range. End-to-end simulation from raw echo samples through 
the 2D-FFT/Keystone chain is deferred to implementation-stage 
validation.

\begin{figure*}[!t]
\centering
\includegraphics[width=0.92\textwidth]{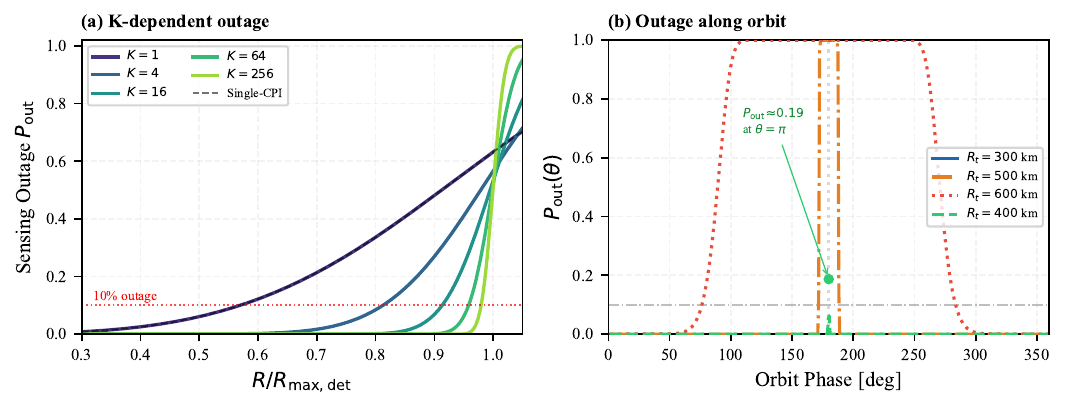}
\caption{Sensing outage under Swerling~I fluctuation
(Proposition~\ref{prop:outage}), session-optimized regime
($K=K^*(\theta)$ from Proposition~\ref{prop:cpi}).
(a)~$P_\mathrm{out}$ vs.\ normalized range for 
$K\in\{1,4,16,64,256\}$; dashed: single-CPI reference.
(b)~$P_\mathrm{out}(\theta)$ along the NRHO for four 
$R_\mathrm{target}$ values. At $K=16$, outage stays below 
10\% for $R_\mathrm{target}\leq 400$~km outside a 
$\pm 10^\circ$ window around perilune, where the 
ICI-reduced $R_\mathrm{max,det}$ for external threats 
raises $P_\mathrm{out}$ to $\approx 0.19$.}
\label{fig:outage}
\end{figure*}

\subsection{Sensing Outage Under RCS Fluctuation}
\label{sec:outage}

The detection ranges derived in Theorem~\ref{thm:rmax} assume 
the mean RCS $\bar{\sigma}$.
Under Swerling~I fluctuation, the instantaneous RCS 
$\sigma \sim \mathrm{Exp}(\bar{\sigma})$ causes 
$R_\mathrm{max}$ to vary from realization to realization.
We define the \emph{sensing outage probability} as the 
probability that the instantaneous detection range falls below 
a specified requirement:

\begin{proposition}[K-CPI Sensing Outage]
\label{prop:outage}
Under Swerling~I RCS fluctuation and non-coherent integration 
over $K$ independent CPIs, the sensing outage probability at 
operational range $R_\mathrm{target}$ and orbital phase 
$\theta$ is
\begin{equation}
P_\mathrm{out}(\theta, K) = 
F_{\Gamma}\!\left(K\cdot\!\left(\frac{R_\mathrm{target}(\theta)}
{R_\mathrm{max,det}(\theta)}\right)^{\!4};\, K,\, 1\right),
\label{eq:outage}
\end{equation}
where $R_\mathrm{max,det}(\theta)$ is the deterministic 
detection range from Theorem~\ref{thm:rmax} computed with 
the mean RCS $\bar{\sigma}$, and 
$F_{\Gamma}(\cdot;\,k,\,1)$ denotes the CDF of the 
Gamma distribution with shape parameter $k$ and unit scale 
(scale parameter fixed at unity throughout this section; 
only the shape parameter $K$ varies).
\end{proposition}

\begin{IEEEproof}
Under Swerling~I, the per-CPI RCS satisfies 
$\sigma_k \sim \mathrm{Exp}(\bar{\sigma})$ independently 
across $k = 1,\ldots,K$. Define the per-CPI nominal SNR 
$\bar{\gamma}(R) \propto R^{-4}$ at the mean RCS; the 
realized per-CPI SNR is then 
$\gamma_k(R) = \bar{\gamma}(R)\,(\sigma_k/\bar{\sigma})$. 
After $K$-pulse non-coherent integration, the test 
statistic is
\begin{equation*}
\Gamma_K(R) = \sum_{k=1}^K \gamma_k(R) = \bar{\gamma}(R)\cdot S_K, 
\quad S_K \triangleq \sum_{k=1}^K \frac{\sigma_k}{\bar{\sigma}}.
\end{equation*}
As a sum of $K$ i.i.d.\ unit-mean exponentials, 
$S_K \sim \mathrm{Gamma}(K, 1)$. The deterministic detection 
range satisfies 
$K\bar{\gamma}(R_\mathrm{max,det}) = \gamma_\mathrm{th}^{(K)}$, 
where $\gamma_\mathrm{th}^{(K)}$ is the K-pulse integration 
threshold from \eqref{eq:rmax}. Substituting 
$\bar{\gamma}(R_\mathrm{target})/\bar{\gamma}(R_\mathrm{max,det}) 
= (R_\mathrm{max,det}/R_\mathrm{target})^{-4}$ and rearranging, 
the outage event 
$\{\Gamma_K(R_\mathrm{target}) < \gamma_\mathrm{th}^{(K)}\}$ 
is equivalent to 
$\{S_K < K(R_\mathrm{target}/R_\mathrm{max,det})^4\}$, 
yielding \eqref{eq:outage}.
\end{IEEEproof}

\begin{remark}[Limiting Cases]
\label{rem:outage_limits}
For $K = 1$, \eqref{eq:outage} reduces to the single-CPI 
exponential form 
$P_\mathrm{out} = 1 - \exp(-(R_\mathrm{target}/R_\mathrm{max,det})^4)$, 
recovering the classical Swerling~I single-pulse result 
\cite{Richards2014}. As $K \to \infty$, the law of large 
numbers gives $S_K/K \xrightarrow{a.s.} 1$, so 
$P_\mathrm{out}$ converges to a unit step at 
$R_\mathrm{target} = R_\mathrm{max,det}$, recovering the 
deterministic limit. The relative variance 
$\mathrm{Var}(S_K/K) = 1/K$ explains the rapid 
convergence: even modest values of $K$ yield 
near-deterministic behavior, as quantified in 
Corollary~\ref{cor:outage}.
\end{remark}

\begin{corollary}[K-Dependent Reliable Range]
\label{cor:outage}
At outage tolerance $\epsilon$, the operationally reliable 
range $R_\epsilon(K)$ satisfies
\begin{equation}
\frac{R_\epsilon(K)}{R_\mathrm{max,det}} = 
\left(\frac{F_\Gamma^{-1}(\epsilon;\,K,\,1)}{K}\right)^{\!1/4}.
\label{eq:R_epsilon}
\end{equation}
Numerical inversion at $\epsilon = 0.1$ gives 
$R_{0.1}/R_\mathrm{max,det} \in 
\{0.57, 0.81, 0.91, 0.96, 0.98\}$ for 
$K \in \{1, 4, 16, 64, 256\}$. For Gateway co-moving debris, 
where $R_\mathrm{max,det}(\theta) \approx 700$~km is nearly 
uniform across the orbit (Fig.~\ref{fig:rmax_warning}(a)), 
$K = 16$ keeps $P_\mathrm{out}(\theta) < 0.10$ for 
$R_\mathrm{target} \leq 600$~km at every phase. For external 
threats, the perilune dip in $R_\mathrm{max,det}$ to 
$\approx 420$~km (Fig.~\ref{fig:coverage}(b)) dominates the 
outage behavior: at $R_\mathrm{target} = 400$~km the 
normalized range $R_\mathrm{target}/R_\mathrm{max,det}$ 
approaches $0.95$ near $\theta = \pi$, raising 
$P_\mathrm{out}$ to $\approx 0.19$, while at 
$R_\mathrm{target} = 500$~km the requirement exceeds 
$R_\mathrm{max,det}$ entirely and outage saturates at unity 
in a narrow window around perilune, as shown in 
Fig.~\ref{fig:outage}(b).
\end{corollary}

Fig.~\ref{fig:outage}(a) plots $P_\mathrm{out}$ versus 
normalized range $R/R_\mathrm{max,det}$ for 
$K \in \{1, 4, 16, 64, 256\}$, with the dashed reference 
showing the single-CPI form. The $K=1$ curve matches the 
reference, confirming the limiting case in 
Remark~\ref{rem:outage_limits}. As $K$ increases, the 
curves steepen and shift toward 
$R/R_\mathrm{max,det} = 1$, reflecting the $1/\sqrt{K}$ 
variance reduction of $S_K/K$.
Fig.~\ref{fig:outage}(b) traces $P_\mathrm{out}(\theta)$ 
along the NRHO orbit for four operational range requirements 
at $K = K^*(\theta)$ from Proposition~\ref{prop:cpi}. 
At the nominal operating point $K \approx 16$, outage 
remains below 10\% for $R_\mathrm{target} \leq 400$~km 
across the apolune-side majority of the orbit; the only 
violation occurs in a narrow $\pm 10^\circ$ window around 
perilune, where the ICI-reduced $R_\mathrm{max,det}$ for 
external threats brings the normalized range close to unity 
and raises $P_\mathrm{out}$ to $\approx 0.19$ at 
$R_\mathrm{target} = 400$~km. This perilune-localized 
degradation provides quantitative justification for the 
risk-weighted requirement profile of 
Section~\ref{sec:adaptive_results}, which deliberately 
relaxes $R_\mathrm{target}(\theta)$ near perilune where the 
encounter density is also lowest. The 
$R_\mathrm{target} = 500$~km curve spikes to near-unity 
at perilune because the requirement exceeds $R_\mathrm{max,det}$ 
entirely there, making it infeasible regardless of $K$ 
(Fig.~\ref{fig:coverage}(b)). The central operational 
implication of Corollary~\ref{cor:outage} is that, at 
integration depths relevant to cislunar opportunistic 
ISAC and away from the perilune ICI dip, RCS fluctuation 
costs less than 10\% of the deterministic detection 
range---approximately a factor of five less than the 
$43\%$ loss obtained from the $K = 1$ case.

\section{Conclusion}
\label{sec:conclusion}

This paper proposed CisLunarSense, an opportunistic ISAC 
framework tailored to the cislunar space environment.
The orbit-phase-dependent CRB reveals that NRHO velocity 
structure couples with ICI to create two operating regimes: 
operational debris ($v_\mathrm{rel} < 50$~m/s, detectable within 
700~km, $>$30~min warning) and external threats 
($v_\mathrm{rel}$ up to 500~m/s, 400--630~km).
The orbit-phase-adaptive allocation reduces the duty cycle 
from 60\% to 19\% and recovers 46~Mbps of throughput 
(44$\to$90~Mbps); $\sim$82\% of the gain comes from 
risk-weighted per-phase $R_\mathrm{target}(\theta)$ 
specification and the remaining 18\% from phase-by-phase 
adaptation unique to eccentric orbits 
(Fig.~\ref{fig:adaptive_rho}).
The $K$-CPI sensing outage probability 
(Proposition~\ref{prop:outage}) shows that, at the nominal 
operating point $K = 16$, RCS fluctuation costs less than 10\% 
of the deterministic detection range, confirming that 
fluctuation is not the binding design constraint at 
operationally relevant integration depths.
The framework scales with the evolving cislunar threat 
environment without dedicated radar hardware, requiring only 
multicarrier-capable relay transceivers.
Future work includes multistatic configurations, solar radiation 
pressure for long-term trajectory prediction, 
track-while-scan orbit determination, and comparison with 
alternative sensing waveforms (e.g., FMCW, OTFS) on the same 
RF chain.

\bibliographystyle{IEEEtran}
\bibliography{refs}

\end{document}